\documentclass[11pt]{article}

\usepackage{amsmath,amsthm,color,multirow}
\usepackage{mathtools}
\usepackage{mathrsfs}
\usepackage{graphicx}
\usepackage{amssymb}
\usepackage{dsfont}
\usepackage{a4}
\usepackage{url}
\usepackage{hyperref}
\usepackage{authblk}

\newcommand{\bea}{\begin{eqnarray*}}
\newcommand{\eea}{\end{eqnarray*}}
\newcommand{\be}{\begin{eqnarray}}
\newcommand{\ee}{\end{eqnarray}}

\def\dim{{\rm {dim}}}

\def\1b{{\mathbf 1}}

\def\0b{{\mathbf 0}}

\newtheorem{rmk}{Remark}
\newtheorem{thm}{Theorem}
\newtheorem{lemma}{Lemma}
\newtheorem{propo}{Proposition}
\newtheorem{defn}{Definition}
\newtheorem{coro}{Corollary}

\theoremstyle{definition}
\newtheorem{ex}{Example}

\begin{document}

\title{Circuits for robust designs}

\author[1]{Roberto Fontana}
\author[2]{Fabio Rapallo}
\author[3]{Henry P. Wynn}
\affil[1]{Department of Mathematical Science, Politecnico di Torino, Italy\footnote{{\tt roberto.fontana@polito.it}}}
\affil[2]{Department of Economics, University of Genova, Italy\footnote{{\tt fabio.rapallo@unige.it}}}
\affil[3]{London School of Economics, UK\footnote{{\tt H.Wynn@lse.ac.uk}}}

\date{}

\maketitle

\begin{abstract}
This paper continues the application of circuit theory to experimental design started by the first two authors. The theory gives a very special and detailed representation of the kernel of the design model matrix. This representation turns out to be an appropriate way to study the optimality criteria referred to as robustness: the sensitivity of the design to the removal of design points. Many examples are given, from classical combinatorial designs to two-level factorial design including interactions. The complexity of the circuit representations are useful because the large range of options they offer, but conversely require the use of dedicated software. Suggestions for speed improvement are made.

\medskip

\end{abstract}

\section{Introduction}
In Design of Experiments, Fractional Factorial Designs are frequently used in many fields of application, including medicine, engineering and agriculture. They offer a valuable tool for dealing with problems where there are many factors involved and each run is expensive. The literature on the subject is extremely rich. A non-exhaustive list of references includes \cite{mukerjee2007modern}, \cite{dey2009fractional}, \cite{hedayat2012orthogonal}, \cite{bailey2008design}.

When searching for an \emph{optimal} experimental designs, we aim to select a design in order to produce the best estimates of the relevant parameters for a given sample size. The are many criteria for choosing an optimal design for the problem under study. A possible classification divides such criteria into two classes: \emph{model-free} and \emph{model-based} criteria. An example of model-free criterion is the minimization of the size of an orthogonal array of a given strength. Examples of model-based criteria include alphabetical design criteria (among these $D$-optimality is one of the most commonly used in applications).

In this work we focus on the model-based setting and we consider the notion of \emph{robustness} of a design. Although most of the examples will concern Fractional Factorial Designs for linear models, the theory developed in this paper is quite general, and some pointers and an example regarding polynomial models are described in Sect.~\ref{sect:circandrob}. The notion of robustness is particularly important when at the end of the experimental activity the design may be incomplete i.e. the response values are not available for all the points of the design itself.  Fractional Factorial Designs with removed runs are studied in, e.g. \cite{butler07}, \cite{street}, \cite{xampeny} and with combinatorial analysis in \cite{fontana|rapallo:19}, but in a model-free context.
Following Ghosh \cite{ghosh:79}, \cite{ghosh:82}, \cite{dee:93}, we consider robustness in terms of the estimability of a given model on the basis of incomplete designs.  We will see that this definition of robustness is also related, but not equivalent, to $D$-optimality.

The model-design pair determines the design matrix. We study a combinatorial object derived from the design matrix  namely the \emph{circuit basis} of the design matrix. We analyze the behavior of the circuits of the design matrix for sub-matrices. This results in sub-fractions (when removing rows) or in super-models (when removing columns). A connection between the circuits and the estimability of saturated designs has been investigated in \cite{fontana|etal:14}.

From the results of this analysis we derive a greedy algorithm to find robust designs by checking the intersections of the design with the supports of the circuits. Moreover, we define a simplified version of this algorithm based on a subset of circuits that can be theoretically characterized (i.e., without symbolic computation) for most cases of factorial designs. The advantage of the simplified version of the algorithm is that it can work in higher dimensions that the standard algorithm. We perform a simulation study where several examples are illustrated to prove the effectiveness of the algorithm. The use of greedy algorithms for $D$-optimality, but in a different context, have been used recently also in \cite{harman:20a} and \cite{harman:20b}.

The paper is organized as follows. In Section \ref{sect:definition} we define the robustness and we obtain its relation with $D$-optimality in the case of totally unimodular design matrices. Some definitions and a few basic facts concerning the algebraic and combinatorial properties of the model matrices are presented in Section \ref{sect:algebraic}. In Section \ref{sect:circandrob} we study the connections between the robustness of a fraction and the circuits contained in the fraction itself. The algorithms and the simulation study are described in Section \ref{sect:algo}. Final remarks are in Section \ref{sect:final}.

\section{Factorial designs, optimal design, and robustness} \label{sect:definition}

In this section, we summarize the basic definitions and notation about the notion of optimality of fractional factorial designs and introduce the definition of robustness.

We adopt here a \emph{candidate set} approach to experimental design. Let ${\mathcal D}$ be a large discrete set in ${\mathbb R}^m$ from which a small set ${\mathcal F}$, usually referred to as design or fraction, is to be selected. We will also consider the general case of fractions with replicates. In this case the fraction ${\mathcal F}$ is a multiset and $\mathcal{D}$ is its underlying set.

One standard example is to consider as candidate set ${\mathcal D}$ the cartesian product of the level sets of the $m$ factors. The set ${\mathcal D}$, when thought of as a design in its own right, is referred to as a full factorial.

We point out that in our theory the coding of the level set is irrelevant, so that for the level set of a factor with $s$ levels we can use  $\{0,\ldots, s-1\}$ or the complex coding
\[
\left\{ \exp \left( \frac {2 \pi i k} {s} \right) \ : k=0,\ldots, s-1 \right\} \,
\]
or any other coding that is considered appropriate. For binary factors ($s=2$) we observe that the complex coding corresponds to choose $\{-1,1\}$ as level sets.

In general, the problem of finding \emph{optimal designs} can be stated as follows. Given a \emph{big set} ${\mathcal D}$ with $K$ design points and a linear model on ${\mathcal D}$, choose an \emph{optimal small set} ${\mathcal F}$ with $n$ design points. Let us denote with ${\bf y}$ the vector of the response variable. We consider a linear model on ${\mathcal D}$ of the form
\begin{equation} \label{classmod}
{\bf y} = X_{\mathcal D} \boldsymbol{\beta} + \boldsymbol{\varepsilon} \, ,
\end{equation}
where $X_{\mathcal D}$ is the full-design model matrix with dimensions $K \times p$, $\boldsymbol{\beta}$ is the $p$-dimensional vector of the parameters, ${\mathbb E}({\bf y}) = X_{\mathcal D} \boldsymbol{\beta}$, and the usual assumptions on the variance hold: ${\mathbb V}({\bf y}) = \sigma^2 I_K$, where $I_K$ is the identity matrix with dimension $N$. Moreover, we assume that the full-design model matrix $X_{\mathcal D}$ is full rank, i.e. its rank is $p$. Although this last assumption is not strictly necessary for the validity of our results, nevertheless it  makes easier theorem statements and proofs, so we work under the full-rank assumption without loss of generality.

For instance, in a two-factor design, first factor $A_1$ with level set $\{0,\ldots , s_1-1\}$ and second factor $A_2$ with level set $\{0,\ldots, s_2-1\}$, under the simple effect model we have $p=s_1+s_2-1$ and a possible design matrix is:
\begin{equation}   \label{mat-repr}
X_{\mathcal D}= \left( {\bf m}_0 \ | \ {\bf a}_0 \ | \ \ldots \ | \ {\bf a}_{s_1-2} \ | \ {\bf b}_0 \ | \
\ldots \ | \ {\bf b}_{s_2-2}  \right) \, ,
\end{equation}
where ${\bf m}_0$ is a column vector of $1$'s, ${\bf a}_0, \ldots, {\bf a}_{s_1-2}$
are the indicator vectors of the first $(s_1-1)$ levels of the
factor $A_1$, and ${\bf b}_0, \ldots, {\bf b}_{s_2-2}$ are the indicator vectors
of the first $(s_2-1)$ levels of the factor $A_2$.

Let us now reconsider the model in Eq.~\eqref{classmod} under the point of view of Polynomial Algebra. From the design ${\mathcal F}$ and given a statistical model with $p$ parameters, we can write the design matrix $X_{\mathcal F}$. A model here is typically a polynomial function
\begin{equation}\label{polyeta}
\boldsymbol{\eta} = {\mathbb E}({\bf y}) = \sum_{\boldsymbol{\alpha}  \in L} c_{\boldsymbol{\alpha}} {\bf x}^{\boldsymbol{\alpha}},
\end{equation}
where we have used the monomial notation: $ \alpha = (\alpha_1, \ldots, \alpha_m)$ and   $x_1^{\alpha_1} x_2^{\alpha_2} \cdots    x_m ^{\alpha_m}$ for a point ${\bf x} = (x_1,x_2,\ldots, x_m)$ in ${\mathbb R}^m$.

The notation $L$, meaning a list of integer exponents, is a convenient way to summarize the model and we shall refer to the models basis $\{x^{\alpha}, \alpha \in L \}$. The design-model pair $({\mathcal F},L)$ gives a design matrix
\begin{equation}\label{polymat}
X_{\mathcal F} = \{ {\bf x}^{\boldsymbol{\alpha}}\}_{{\bf x} \in {\mathcal F}, \boldsymbol{\alpha} \in L } \, .
\end{equation}

Notice that the definition in Eqs.~\eqref{polyeta} and \eqref{polymat} still holds when ${\mathcal F} = {\mathcal D}$, and that for a fixed design the matrix $X_{\mathcal F}$ is obtained from ${\mathcal D}$ simply by selection of the appropriate rows of ${\mathcal D}$. In case ${\mathcal F}$ contains replicates  the matrix $X_{\mathcal F}$ is obtained from ${\mathcal D}$ by replication of the appropriate rows of ${\mathcal D}$.

In the model-based approach to experimental design, the quality of the chosen design ${\mathcal F}$ is expressed by some properties of $X_{\mathcal F}$. We shall be particulary interested in the D-optimality of the design ${\mathcal F}$, as a proper subset or eventually as a subset with replicates, of the candidate set $\mathcal D$. As described in \cite{sas2004sas}, D-optimality is based on the determinant of the information matrix for the design, which is the same as the reciprocal of the determinant of the variance-covariance matrix for the least squares estimates of the linear parameters of the model. The D-efficiency of a design $\mathcal{F}$ with design matrix $X_\mathcal{F}$ is defined as
\begin{equation} \label{defficiency}
D(X_\mathcal{F}) =100\times \left( \frac{ \det(X_\mathcal{F}^t X_\mathcal{F})^{1/p}}{n} \right)
 \end{equation}
where $p$ is the number of parameters in the linear model, and $n$ is the number of design points, $n=\#\mathcal{F}$. The D-efficiency is the relative number of runs (expressed as percentages) that are required by a hypothetical orthogonal design to achieve the same  $\det(X_\mathcal{F}^t X_\mathcal{F})$, \cite{mitchell1974computer}.

Given  a design $\mathcal{F}$ we will compare its D-optimality with its \emph{robustness}. For a design $\mathcal{F}$ to be full rank we must have $n \geq p$. Let us suppose that $n > p$ but that for some unexpected reasons $n-p$ points of $\mathcal{F}$ are lost. We would obtain a size $p$ design that we denote by $\mathcal{F}_p$. The corresponding design matrix $X_{\mathcal{F}_p}$ could be full rank and then the $p$ parameters are estimable or not. When $X_{\mathcal{F}_p}$ is full rank the corresponding design $\mathcal{F}_p$ is said to be saturated. We define the robustness of a fraction $\mathcal{F}$ as the ratio between the number of saturated fractions and the total number of fractions of size $p$ which are obtained by removing $n-p$ points from $\mathcal{F}$.

\begin{defn}[Robustness]
We define the {\emph robustness} of a fraction ${\mathcal F}$  with design matrix $X_{\mathcal F} $as
\begin{equation} \label{def:robustness}
r(X_{\mathcal F}) = \frac {\# \mbox{ saturated } \ {\mathcal F}_p } {\# {\mathcal F}_p}= \frac {\# \mbox{ saturated } \ {\mathcal F}_p } {{n \choose p}} \, .
\end{equation}
\end{defn}

In particular a design $\mathcal{F}$ is \emph{robust} if its robustness is equal $1$, $r(X_{\mathcal F}) =1$.
Now we recall the definition of totally unimodular matrices and prove that for totally unimodular matrices $X_{\mathcal F}$ the definitions of robustness in Eq.~\eqref{def:robustness} and D-efficiency in Eq.~\eqref{defficiency} are equivalent. This follows from the Cauchy-Binet lemma from Linear Algebra.

\begin{lemma}[Cauchy-Binet]
Let $X_p$ be a $p \times p$ sub-matrix of $X$. Then:
\begin{equation*}
\det(X^t X) = \sum \det(X_p^t X_p) = \sum \det(X_p)^2 \, ,
\end{equation*}
where the sum extends over all $p \times p$ sub-matrices of $X$.
\end{lemma}

\begin{defn}
A matrix $X$ is totally unimodular if every square sub-matrix $X_p$ has determinant $0$, $+1$, or $-1$. In particular, this implies that all entries are $0$ or $\pm 1$.
\end{defn}

\begin{propo} \label{deffrob}
Let us consider a fraction ${\mathcal F}$ with design matrix $X_{\mathcal F}$. If the design matrix $X_{\mathcal F}$ is totally unimodular then the relation between the robustness $R(X_{\mathcal F})$ and the D-efficiency $D(X_{\mathcal F})$ is
\begin{eqnarray*}
r(X_{\mathcal F}) = \frac{\left( n  D(X_{\mathcal F}) /100\right)^p} {{n \choose p}} \, ; \\
D(X_{\mathcal F}) =  \frac{100 \left({n \choose p}  r(X_{\mathcal F}) \right)^{1/p} }{n} \, .
\end{eqnarray*}
and in particular $D$-optimality is equivalent to maximum robustness.
\end{propo}
\begin{proof}
Using the Cauchy-Binet lemma we can write
\begin{equation*}
\det(X_{\mathcal F}^t X_{\mathcal F} ) = \sum \det(X_p)^2 \, ,
\end{equation*}
where the summation is extended to all the $p \times p$ sub-matrices  $X_p$ of $X_{\mathcal F}$.

For totally unimodular matrices  $\det(X_p)^2 \in \{0,1\}$ and then
$\sum \det(X_p)^2$ is the number of saturated fractions $\mathcal{F}_p$ contained in ${\mathcal F}$. It follows that the robustness of the fraction ${\mathcal F}$ with design matrix $X_{\mathcal F}$ can be written as
\[
r(X_{\mathcal F}) =  \frac {\det(X_{\mathcal F}^t X_{\mathcal F} )} {{n \choose p}}.
\]
Finally, from the definition of $D$-efficiency in Eq. \eqref{defficiency}, the result follows.
\end{proof}

\begin{ex}[Balanced Incomplete Block Design] \label{ex:BIBD1}
Let us consider an example with two factors, $A$ with level set $\{1,2,3,4\}$ and $B$ with level set $\{1,2,3,4,5,6\}$. As candidate set $\mathcal{D}$ we choose an Orthogonal Array of size $12$ and strength $1$. If the levels of $A$ ($B$) represent the rows (the columns) of a table we can represent the points of $\mathcal{D}$ as the \emph{bullets} in Fig.~\ref{fig:bibd}. The design $\mathcal{D}$ can also be seen as a Balanced Incomplete Block Design (BIBD) with $t=4$ treatments, $b=6$ blocks, $k=2$ treatments in each block and $\lambda=1$ (i.e. each pair of treatments occurs together $\lambda=1$ time within a block).

\begin{figure}
\begin{center}
\begin{tabular}{|c|c|c|c|c|c|} \hline
$\bullet$ & $\bullet$ & $\bullet$ & & &  \\ \hline
$\bullet$ & & & $\bullet$ & $\bullet$ &  \\ \hline
& $\bullet$ & & $\bullet$ & & $\bullet$  \\ \hline
& & $\bullet$ & & $\bullet$ & $\bullet$  \\ \hline
\end{tabular}
\caption{The candidate set $\mathcal{D}$ for Example \ref{ex:BIBD1}.} \label{fig:bibd}
\end{center}
\end{figure}
One possible full rank design matrix for the main effects model is
\begin{equation}   \label{xd}
X_{\mathcal D}= \left( {\bf m}_0 \ | \ {\bf a}_1 \ | \ \ldots \ | \ {\bf a}_{3} \ | \ {\bf b}_1 \ | \
\ldots \ | \ {\bf b}_{5}  \right) \, ,
\end{equation}
where ${\bf m}_0$ is a column vector of $1$'s, ${\bf a}_1, \ldots, {\bf a}_{3}$
are the indicator vectors of the first $3$ levels of the
factor $A$, and ${\bf b}_1, \ldots, {\bf b}_{5}$ are the indicator vectors
of the first $5$ levels of the factor $B$. The rank of the matrix $X_{\mathcal D}$ is $p=9$.

It can be proved that $X_{\mathcal D}$ as defined in Eq. \eqref{xd} is totally unimodular, \cite{schrijver:86}. It follows that any sub-matrix $X_\mathcal{F}$ of $X_{\mathcal D}$ obtained by selecting $n$ rows from $X_{\mathcal D}$ with $n=9,10,11$ is totally unimodular. Then from Proposition \ref{deffrob} robustness and D-efficiency are equivalent for all the fractions of $\mathcal{D}$.

Let us suppose that we want to find robust sub-fractions with $n=10$ runs. By simply checking all the ${12 \choose 10}=66$ fractions of ${\mathcal D}$ we find $6$ fractions that do not allow estimability of the model (the rank of the design matrix is less than $p=9$), $48$ fractions with robustness equal to $0.6$ and $12$ fractions with robustness equal to $0.8$.
\end{ex}

Except from the special case of totally unimodular matrices, robustness and D-efficiency are not related as in Proposition \ref{deffrob}. In the next sections we will explore the relationship between robustness and D-efficiency in different scenarios.

\section{Circuits and their properties} \label{sect:algebraic}

In order to give a complete account of our theory and to present our algorithms with full details, some definitions and a few basic facts concerning the algebraic and combinatorial properties of the model matrices are needed. Thus, in the first part of this section, we recall some definitions and results from Algebraic Statistics. The interested reader can find a detailed presentation in \cite{pistone|etal:01}. As a general reference for Commutative Algebra we refer to \cite{cox|etal:07}.

Let us consider a design matrix on a set of $K$ design points. For instance, such a set can be the full factorial design, ${\mathcal D}$, and in this case $K = N$, but the theory is not limited to full factorial designs. Let $X=X_{\mathcal F}$ be a model matrix on ${\mathcal F}$, and assume that $X$ has integer entries. To simplify the notation, we drop the subscript ${\mathcal F}$ if there is no ambiguity. The matrix $X$ has dimensions $K \times p$. Moreover, in order to match the common notation in Statistics with the notation in Commutative Algebra, we consider the matrix $A=X^T$, the transpose of the model matrix.

Given a $p \times K$ integer matrix $A$, we define the {\emph polynomial ring} ${\mathbb R}[\mathbf{x}] = {\mathbb R}[x_1, \ldots, x_K]$ of all polynomials with indeterminates $x_1,\ldots, x_K$ with real coefficients, i.e., we define an indeterminate for each element of ${\bf y}$ or, equivalently, for each point of the design. An ideal ${\mathcal I}$ in ${\mathbb R}[\mathbf{x}]$ is a subset of ${\mathbb R}[\mathbf{x}]$ such that $f+g \in {\mathcal I}$ for all $f,g \in {\mathcal I}$ and $fg \in {\mathcal I}$ for all $f \in {\mathcal I}$ and for all $g \in {\mathbb R}[\mathbf{x}]$. The ideal generated by the polynomials $f_1, \ldots , f_r$ is the ideal
\[
{\mathcal I}(f_1, \ldots, f_r) = \langle f_1, \ldots, f_r \rangle = \{g_1f_1 + \ldots g_rf_r \ : \ g_1, \ldots, g_r \in {\mathbb R}[\mathbf{x}] \} \, .
\]
A classical result in polynomial algebra, namely the Hilbert basis theorem, ensures that every ideal in ${\mathbb R}[\mathbf{x}]$ is finitely generated.

The \emph{toric ideal} defined by $A$ is the binomial ideal (i.e., an ideal generated by binomials)
\begin{equation*}
{\mathcal I}(A) = \langle \mathbf{x}^\mathbf{a} - \mathbf{x}^\mathbf{b} \ : \ A\mathbf{a} = A\mathbf{b} \rangle
\end{equation*}
where the monomials $\mathbf{x}^\mathbf{a}$ are written in vector notation $\mathbf{x}^\mathbf{a}= x_1^{a_1} \cdots x_K^{a_K}$.

%

In our examples we label the design points lexicographically for convenience. For instance in the $2^4$ case we define the indeterminates as
follows: $(-1,-1,-1,-1) \mapsto x_1$, $(-1,-1,-1,1) \mapsto x_2$, $(-1,-1,1,-1) \mapsto x_3$ and so on until $(1,1,1,1) \mapsto x_{16}$. Moreover, we use the log notation for binomials:
\begin{equation*}
f = \mathbf{x}^\mathbf{a}-\mathbf{x}^\mathbf{b} \longmapsto  \mathbf{a} - \mathbf{b}
\end{equation*}
when this helps in simplifying the presentation.

\begin{defn}
The {\emph support} of a binomial $f = \mathbf{x}^\mathbf{a}-\mathbf{x}^\mathbf{b}$ is the set of indices $i$
($i=1, \ldots , K$) such that $a(i) \ne 0$ or $b(i) \ne 0$. We denote the  support of $f$ with $\mathrm{supp}({f})$.
\end{defn}

\begin{defn}
An irreducible binomial $f = \mathbf{x}^\mathbf{a}-\mathbf{x}^\mathbf{b} \in {\mathcal I}(A)$ is a
{\emph circuit} if there is no other binomial $g \in {\mathcal I}(A)$ such
that ${\rm supp}({g}) \subset {\rm supp}({f})$ and ${\rm supp}({g}) \ne {\rm supp}({f})$. We
denote the set of all circuits of ${\mathcal I}(A)$ with ${\mathcal
C}(A)$.
\end{defn}

\begin{rmk}
Each column of $A$ identifies a design point, and therefore the
definition of a set of column-indices is equivalent to the
definition of the fraction with the corresponding design points.
Given ${\mathcal F}=\left\{i_1, \ldots, i_n \right\}$,
$A_{\mathcal F}$ is the sub-matrix of $A$ obtained by selecting the
columns of $A$ according to ${\mathcal F}$.
\end{rmk}

The set ${\mathcal C}(A)$ is also called the circuit basis of the matrix $A$. 
With a slight abuse of notation, we denote with ${\mathcal C}(A)$ also the set of the exponents, i.e.,
${\mathcal C}(A) := \{\mathbf{u} \ | \ \mathbf{x}^{\mathbf{u}+}-\mathbf{x}^{\mathbf{u}-} \ \textrm{ is a circuit }\}$
and we call such integer vectors the circuits of ${\mathcal F}$ with respect to $A$.

Notice that by construction a circuit $\mathbf{u}$ must belong to $\ker(X^t)$.

In the following proposition we collect some major properties of the circuits. The proofs can be found in \cite{sturmfels:96}.

\begin{propo}  \label{prop:varie}
\begin{enumerate}
\item Every circuit is a primitive binomial, i.e., if $\mathbf{x}^{\mathbf{u}+}-\mathbf{x}^{\mathbf{u}-}$ is a circuit, then there is no binomial $\mathbf{x}^{\mathbf{ v}+}-\mathbf{x}^{\mathbf{v}-}$ such that $\mathbf{x}^{\mathbf{v}+}$ properly divides $\mathbf{x}^{\mathbf{u}+}$ and $\mathbf{x}^{\mathbf{v}-}$ properly divides $\mathbf{x}^{\mathbf{u}-}$.

\item Every vector $\mathbf{v} \in \ker(X^t)$ can be written as a non-negative rational combination of $(n-p)$ circuits
\[
\mathbf{v} = \sum c_j (\mathbf{x}^{\mathbf{u}+_j} - \mathbf{x}^{\mathbf{u}-_j}) \qquad c_j \in {\mathbb Q}, \  \mathbf{x}^{\mathbf{u}+_j} - \mathbf{x}^{\mathbf{u}-_j} \in {\mathcal C}(X^t) \, .
\]
Each circuit in the previous decomposition is sign-compatible with $\mathbf{v}$.

\item The support of a circuit has cardinality at most $(p+1)$.
%
\end{enumerate}
\end{propo}

\begin{rmk}
The circuit basis of an integer matrix $A$ can be computed through several packages for symbolic computation. The computations presented in the present paper are carried out with {\tt 4ti2}, see \cite{4ti2}. {\tt 4ti2} can be used as an independent executable program or as a package of the Computer Algebra System {\tt Macaulay 2}, see \cite{M2}. For small designs the computations are performed in few seconds at most, and the circuit basis in the output can be easily analyzed.
\end{rmk}

\begin{ex}
First, we illustrate an example in some details. We consider the full factorial $2^4$ design with main effects. Writing the transposed matrix to save space, a full-rank version of the design matrix is
\begin{equation}
X^t = \begin{pmatrix}
1 & 1 & 1 & 1 & 1 & 1 & 1 & 1 & 1 & 1 & 1 & 1 & 1 & 1 & 1 & 1 \\
1 & 1 & 1 & 1 & 1 & 1 & 1 & 1 & 0 & 0 & 0 & 0 & 0 & 0 & 0 & 0 \\
1 & 1 & 1 & 1 & 0 & 0 & 0 & 0 & 1 & 1 & 1 & 1 & 0 & 0 & 0 & 0 \\
1 & 1 & 0 & 0 & 1 & 1 & 0 & 0 & 1 & 1 & 0 & 0 & 1 & 1 & 0 & 0 \\
1 & 0 & 1 & 0 & 1 & 0 & 1 & 0 & 1 & 0 & 1 & 0 & 1 & 0 & 1 & 0 \\
\end{pmatrix} \, ,
\end{equation}
where each column is a design point (lexicographically from $(-1,-1,-1,-1)$ to $(1,1,1,1)$) and the five rows are the intercept plus one parameter for each main effect. The circuits in ${\mathcal C}(\mathcal D)$ are $1,348$. More precisely, there are:
\begin{itemize}
\item[(a)] $100$ circuits with support on $4$ points: they are of the form
\[
(1, -1, -1,  1,  0,  0,  0,  0,  0 , 0 , 0 , 0 , 0 , 0 , 0,  0)
\]
(such moves are well known in Algebraic Statistics and are named as ``basic moves'' in the context of contingency table analysis).

\item[(b)] $160$ circuits with support on $5$ points: they are of the form
\[
( 1, -2,  0,  1 , 0 , 0 , 0 , 0 , 0 , 0,  0 , 0 , 0,  1 ,-1,  0)\, .
\]

\item[(c)] $1,088$ circuits with support on $6$ points: in this case there are different patterns of nonzero elements. For instance, there are $384$ circuits of the (symmetric) form
    \[
    ( 1, -2,  0 , 2,  0 , 0,  0, -1,  0,  0, -1,  0,  0,  1 , 0,  0 ) \, .
    \]
There are also asymmetric configurations such as the $16$ circuits of the form
\[
( 1,  0 , 0 , 0 , 0 , 0 , 0 ,-1 , 0 , 0,  0, -1,  0 ,-1, -1,  3  ) \, .
\]
\end{itemize}
\end{ex}

\begin{ex}
Let us consider the full factorial $2^4$ design with main effects and second order interactions. A full rank version of the design matrix has dimensions $16 \times 11$ and there are $140$ circuits: $20$ circuits with support on $8$ points; $40$ circuits with support on $10$ points; $80$ circuits with support on $12$ points.
\end{ex}

Other examples, also with multi-level and asymmetric designs, include:
\begin{itemize}
\item[(i)] Design $2^5$; model with simple factors, 2-way and 3-way interactions. There are $3,254$ circuits that can be divided into $12$ classes, up to permutations of factors or levels.

\item[(ii)] Design $2^5$; model with simple factors. The circuits are $353,616$ and they can be divided into $38$ classes.

\item[(iii)] Design $2\times 3 \times 4$; model with simple factors and 2-way interactions. There are $42$ circuits that
can be divided into $2$ classes.

\item[(iv)] Design $3\times 3 \times 4$; model with simple factors and 2-way interactions. There are $19,722$ circuits
that can be divided into $20$ classes.
\end{itemize}

Note that all the computations for the Examples above can be performed with ${\tt 4ti2}$ in less than 1 second, but the computational cost (and the number of circuits) increases fast with the dimension of the full factorial design. For instance, there are $353,616$ circuits for the full-factorial $2^5$ design with main effects, and the computation takes about 8 hours of CPU time. This makes important some properties of the circuits that we state and discuss in the next section. These properties will allow computations also for medium-sized designs, where the computations on the full factorial design are unfeasible.

A first connection between the circuits and the properties of the designs has been presented in \cite{fontana|etal:14} and concerns saturated designs.
A design ${\mathcal F}$, subset of a full factorial design ${\mathcal D}$, is a saturated design if it has minimal cardinality $\#{\mathcal D}=p$ and it allows us to estimate the model parameters. Thus, by definition the model matrix $X_{ \mathcal F}$ of a saturated design (under a full-rank parametrization) is a non-singular matrix with dimensions $p \times p$.
The following theorem replaces a linear algebra condition with a combinatorial property for checking whether a design with $p$ runs is
saturated or not.

\begin{thm}
Let $A$ be a (full-rank) full-design model matrix with dimensions $p \times K$
and let ${\mathcal C}_A=\{f_1 , \ldots, f_r\}$ be the set of its
circuits. Given a set ${\mathcal F}$ of $p$ column-indices of $A$, the
sub-matrix $A_{\mathcal F}$ is non-singular if and only if $
{\mathcal F}$ does not contain any of the supports $\mathrm{
supp}(f_1), \ldots , \mathrm{supp}(f_r)$.
\end{thm}

\begin{ex}
Consider the $2^4$ full-factorial design and the model with simple effects and 2-way interactions.
%

The design matrix $X$ has rank $11$, thus we search for fractions with $11$ points. The design
\begin{equation*}
\begin{split}
{\mathcal F}_1 = \{ (0,0,0,0),
(0,0,0,1),(0,0,1,1),(0,1,0,1),(0,1,1,0),(1,0,0,1),\\ (1,0,1,0),
(1,1,0,0),(1,1,0,1),(1,1,1,0),(1,1,1,1) \}
\end{split}
\end{equation*}
is not saturated, but replacing the point $(0,1,0,1)$ with the point $(0,1,0,0)$ we obtain the saturated design
\begin{equation*}
\begin{split}
{\mathcal F}_2 = \{(0,0,0,0),
(0,0,0,1),(0,0,1,1),{(0,1,0,0)},(0,1,1,0),(1,0,0,1),\\ (1,0,1,0),
(1,1,0,0),(1,1,0,1),(1,1,1,0),(1,1,1,1) \} \, .
\end{split}
\end{equation*}
In fact, the full-factorial design matrix has $140$ circuits. They can be divided into three classes,
up to permutations of factors or levels:
\begin{itemize}
\item[(a)] 20 circuits of the form
\begin{equation*}
\mathbf{u}_1=(0,0,0,0,1,-1,-1,1,-1,1,1,-1,0,0,0,0) \, ;
\end{equation*}
\item[(b)] 40 circuits of the form
\begin{equation*}
\mathbf{u}_2=(1,-2,0,1,0,1,-1,0,0,1,-1,0,-1,0,2,-1)\, ;
\end{equation*}
\item[(c)] 80 circuits of the form
\begin{equation*}
\mathbf{u}_3=(1,0,-2,1,0,-1,1,0,-2,1,3,-2,1,0,-2,1)
\end{equation*}
\end{itemize}
and it is immediate to check that the fraction ${\mathcal F}_1$ contains the support of the circuit $\mathbf{u}_2$.
\end{ex}

\section{Robustness and circuits} \label{sect:circandrob}

In this section we study the connections between the robustness of a fraction and the circuits contained in the fraction itself. The first key result here states that the circuits are consistent with the operation of subset selection. This has two relevant consequences. First, it yields a convenient mathematical framework for design search, and we will exploit this fact in the next section, where we will define an algorithm for finding robust fractions. Second, for a given problem of subset selection, the circuit basis can be computed only once on the candidate set and this is enough to perform the analysis on all possible fractions.

Consider two fractions ${\mathcal F}_1$ and ${\mathcal F}_2$ with $k_1$ and $k_2$ design points respectively, such that ${\mathcal F}_1 \subset {\mathcal F}_2$. Without loss of generality, the matrix $X_{{\mathcal F}_2}$ can be partitioned into
\begin{equation*}
X_{{\mathcal F}_2} = \left(
\begin{array}{c}
X_{{\mathcal F}_1} \\
X_{{\mathcal F}_2 - {\mathcal F}_1}
\end{array} \right) \, .
\end{equation*}
and each vector $\mathbf{u} \in {\mathbb Z}^{k_2}$ can be written as
\begin{equation*}
\mathbf{u} = (\mathbf{u}_{{\mathcal F}_1}, \mathbf{u}_{{\mathcal F}_2-{\mathcal F}_1}) \qquad \mbox{ with } \ \mathbf{u}_{{\mathcal F}_1} \in {\mathbb Z}^{k_1}
\end{equation*}

\begin{thm} \label{thm:cirsub}
If ${\mathcal F}_1$ and ${\mathcal F}_2$ are two fractions with ${\mathcal F}_1 \subset {\mathcal F}_2$, then the circuits in ${\mathcal C}({X_{{\mathcal F}_1}})$ are
\begin{equation*}
\{ \mathbf{u}_{{\mathcal F}_1}  \ : \ {u} \in {\mathcal C}({X_{{\mathcal F}_2}}) \ \mbox{ with } \ \mathrm{supp}(\mathbf{u}) \subset {\mathcal F}_1 \} \, .
\end{equation*}
\end{thm}

\begin{proof}
Let $\mathbf{u}$ be a circuit in ${\mathcal C}({X_{{\mathcal F}_2}})$ such that $\mathrm{supp}(\mathbf{u}) \subset {\mathcal F}_1$.
Then $\mathbf{u}_{{\mathcal F}_1} \in \ker(X_{{\mathcal F}_1})^t$ and it is support-minimal.
To prove this, suppose that there is an integer vector $\mathbf{v} \in \ker(X_{{\mathcal F}_1}^t)$ with $\mathrm{supp}(\mathbf{v}) \subset \mathrm{supp}(\mathbf{u})$.
This implies that the vector $\mathbf{u}'=(\mathbf{u}, \mathbf{0}_{{\mathcal F}_2-{\mathcal F}_1})$ obtained by filling with zeros the vector $\mathbf{u}$ to reach the size of ${\mathcal F}_2$ has support containing the support of $\mathbf{v}'=(\mathbf{v}, \mathbf{0}_{{\mathcal F}_2-{\mathcal F}_1})$, and thus $\mathbf{u}'$ is not a circuit in ${\mathcal C}({X_{{\mathcal F}_2}})$.

On the other side, suppose that $\mathbf{v}$ is a circuit in ${\mathcal C}({X_{{\mathcal F}_1}})$. Then it is easy to see that $\mathbf{u}=(\mathbf{v},\mathbf{0}_{{\mathcal F}_2-{\mathcal F}_1})$ is a circuit of ${\mathcal C}({X_{{\mathcal F}_1}})$ and $\mathbf{v}=\mathbf{u}_{{\mathcal F}_1}$.
\end{proof}

In order to show how to apply the previous result, let us consider again the BIBD example already introduced in Sect.~\ref{sect:definition}.

\begin{ex}[BIBD revisited] \label{ex:BIBD2}
Let us consider again the BIBD example with the candidate set ${\mathcal D}$ pictured in Fig.~\ref{fig:bibd}. The model matrix of the complete design has $1,650$ circuits but only $7$ of them have support in our candidate set ${\mathcal D}$. They are listed below:

\begin{equation*}
\begin{matrix*}[r]
    0 & 0 & 0 & 0 & 1 & -1 & 0 & -1 & 1 & 0 & 1 & -1 \\
    0 & 1 & -1 & 0 & 0 & 0 & -1 & 0 & 1 & 1 & 0 & -1 \\
    1 & -1 & 0 & -1 & 1 & 0 & 1 & -1 & 0 & 0 & 0 & 0 \\
    1 & 0 & -1 & -1 & 0 & 1 & 0 & 0 & 0 & 1 & -1 & 0 \\
    0 & 1 & -1 & 0 & -1 & 1 & -1 & 1 & 0 & 1 & -1 & 0 \\
    1 & -1 & 0 & -1 & 0 & 1 & 1 & 0 & -1 & 0 & -1 & 1 \\
    1 & 0 & -1 & -1 & 1 & 0 & 0 & -1 & 1 & 1 & 0 & -1
\end{matrix*}
\end{equation*}

Our problem is to select 10 points defining a sub-fraction with robustness as higher as possible.

Now, note that the best selection strategy would be to remove two points such that all circuits are canceled, i.e., a set of 10 points with no circuits inside. But this can not be done here. Each pair of points we choose to remove preserves at least one circuit. For instance, removing the first and the second points preserves the first circuit, removing the last two points preserves the third circuit, and so on. Thus, at least one circuit survives and the full robustness can not be reached. To choose the best subset we need to inspect with some more details the circuits. We observe that there are 4 circuits with support on 6 points and 3 circuits with support on 8 points. Since the robustness of a fraction depends on the estimability of the saturated sub-fractions, it is better to remove the small circuits as much as possible. In this case, the best solution is obtained by killing all circuits with support on 6 points. This can be done in 12 ways, obtaining 12 sub-fractions which share the same robustness, and actually the maximum achievable in this problem. Such 12 fractions are obtained by removing the following pairs of points:
\[
\{1,9\},\{1,12\},\{2,6\},\{2,11\},\{3,5\},\{3,8\},\{4,9\},\{4,12\},\{5,10\},\{6,7\},\{7,10\},\{8,10\} \, .
\]
\end{ex}

As a general rule, from the definition of robustness and the property of the circuits stated in Theorem \ref{thm:cirsub}, to obtain a robust fraction we need to remove as much circuits as possible, and in particular we need to remove the circuits with small support.

In the example below, we show that the connections between circuits and robustness are not limited to linear models, but they can be used in the general case of polynomial models, which represent the most general class of models on a finite grid of points, according to the expression in Eq.~\eqref{polyeta}. We limit the computation to the univariate case to ease the discussion of the results, but the computations can be extended to a general multivariate polynomial model.

\begin{ex}
Let us consider a (univariate) polynomial model on 7 points $\{-3,-2,-1,0,1,2,3\}$. The model matrix for a quadratic function is
\[
X^t = \begin{pmatrix*}[r]
1 & 1 & 1 & 1 & 1 & 1 & 1 \\
-3 & -2 & -1 & 0 & 1 & 2 & 3 \\
9 & 4 & 1 & 0 & 1 & 4 & 9
\end{pmatrix*} \, .
\]
The circuit basis of $X^t$ is formed by 35 circuits with support on 4 points. The 35 supports cover all possible subsets with 4 points. Thus, in this example all the sub-fractions have the same robustness.
\end{ex}

In view of the last property of the circuits stated in Proposition \ref{prop:varie}, another interesting case happens when we search for fractions with $p+1$ runs.

\begin{coro}
The supports of circuits with $p+1$ points are fully robust fractions.
\end{coro}
\begin{proof}
Since the circuits are support-minimal, if a fraction ${\mathcal F}$ with $p+1$ points is the support of a circuit, then there are no circuits with support contained in ${\mathcal F}$, and in particular all the sub-fractions of ${\mathcal F}$ with $p$ points are estimable.
\end{proof}

\begin{ex}[OA]
Consider the following Orthogonal Array $\mathcal F$ with $18$ runs of strength $2$ for the $2 \times 3^3$ design. This is the best GWLP Orthogonal Array according to the Eendebak catalogue, see \cite{eendebak:sito}. We write the transposed of the design to save space, so each row is a factor.

\begin{equation*}
\begin{matrix}
0 & 0 & 0 & 0 & 0 & 0 & 0 & 0 & 0 & 1 & 1 & 1 & 1 & 1 & 1 & 1 & 1 & 1 \\
0 & 0 & 0 & 1 & 1 & 1 & 2 & 2 & 2 & 0 & 0 & 0 & 1 & 1 & 1 & 2 & 2 & 2 \\
0 & 1 & 2 & 0 & 1 & 2 & 0 & 1 & 2 & 0 & 1 & 2 & 0 & 1 & 2 & 0 & 1 & 2 \\
0 & 1 & 2 & 1 & 2 & 0 & 2 & 0 & 1 & 1 & 2 & 0 & 2 & 0 & 1 & 0 & 1 & 2
\end{matrix}
\end{equation*}

A full-rank version of the (transposed of the) design matrix $X_{\mathcal F}$ for the main effect model on this fraction is
\begin{equation*}
X_{\mathcal F}^t = \begin{pmatrix}
1 & 1 & 1 & 1 & 1 & 1 & 1 & 1 & 1 & 1 & 1 & 1 & 1 & 1 & 1 & 1 & 1 & 1\\
1 & 1 & 1 & 1 & 1 & 1 & 1 & 1 & 1 & 0 & 0 & 0 & 0 & 0 & 0 & 0 & 0 & 0\\
1 & 1 & 1 & 0 & 0 & 0 & 0 & 0 & 0 & 1 & 1 & 1 & 0 & 0 & 0 & 0 & 0 & 0\\
0 & 0 & 0 & 1 & 1 & 1 & 0 & 0 & 0 & 0 & 0 & 0 & 1 & 1 & 1 & 0 & 0 & 0\\
1 & 0 & 0 & 1 & 0 & 0 & 1 & 0 & 0 & 1 & 0 & 0 & 1 & 0 & 0 & 1 & 0 & 0\\
0 & 1 & 0 & 0 & 1 & 0 & 0 & 1 & 0 & 0 & 1 & 0 & 0 & 1 & 0 & 0 & 1 & 0\\
1 & 0 & 0 & 0 & 0 & 1 & 0 & 1 & 0 & 0 & 0 & 1 & 0 & 1 & 0 & 1 & 0 & 0\\
0 & 1 & 0 & 1 & 0 & 0 & 0 & 0 & 1 & 1 & 0 & 0 & 0 & 0 & 1 & 0 & 1 & 0\\
\end{pmatrix} \, ,
\end{equation*}

There are $591$ circuits in ${\mathcal C}({\mathcal X_{\mathcal F}})$, namely: $27$ circuits with support on $4$ points; $114$ circuits with support on $6$ points; $270$ circuits with support on $8$ points; $180$ circuits with support on $9$ points. The $180$ circuits with support on $9$ points are the sub-fractions of ${\mathcal F}$ with 9 points.
\end{ex}

Note that the computation of the circuits in the previous example is easily performed with {\tt 4ti2} in 0.07 seconds while the circuits of the corresponding full factorial design are actually unfeasible. This fact shows once more the relevance of the approach based on the circuits.

In view of the two examples discussed above, two remarks are now in order. First, the circuits with support on $p+1$ points and the circuits with support on $p$ points or less have a completely different role in finding robust fractions. While the circuits with support on $p$ points or less yield non estimable minimal fractions, and therefore they should be avoided as much as possible, the circuits with support on $p+1$ points defines fully robust fractions. Second, when a circuit with support on $p$ points or less is contained in fraction, the loss in robustness it causes is as high as small the support is. In fact, a small circuit will have impact on a large number of minimal fractions, while on the opposite side a circuit with support on $p$ points will produce only one non estimable minimal fraction. Such remarks will be useful in the next section, where an algorithm for finding robust fractions will be introduced.

Finally, we point out that the result on the estimability of saturated fractions mentioned in the previous section comes now as a corollary of Theorem \ref{thm:cirsub}. Moreover, we can state a slight generalization as follows.

\begin{propo}
Consider a fraction ${\mathcal F} \subset {\mathcal D}$ with $k>p$ points. Then the fraction ${\mathcal F}$ is estimable (i.e., the $p$ parameters are estimable) if and only if there is at least one fraction ${\mathcal F}_1$ with $p$ design points that does not contain supports of the circuits in ${\mathcal C}(X_{\mathcal D})$.
\end{propo}

For the proof it is enough to apply \ref{thm:cirsub} and the Cauchy-Binet lemma.

\section{Algorithm for robust fractions} \label{sect:algo}

In this section we describe a simple algorithm for finding robust fractions of a specified size. The basic idea of the algorithm is to improve a given fraction by exchanging, for a certain number of times, the worst point of the fraction with the best point among those which are in the candidate set but not in the fraction.  This kind of algorithms is commonly used in design generation. In general, they are referred to as \emph{exchange} algorithms, see e.g. \cite{wynn1970sequential}.

\subsection{The algorithm and a simulation study}

The basic points of the algorithms come from the theory discussed in the the previous section. First, the a good fraction should avoid as much as possible the circuits with support on $p$ points or less. Second, small circuits are worse than large circuits, since they are contained in a larger number of minimal fractions, leading to a higher loss in robustness, as noticed in Example \ref{ex:BIBD2}. Thus, at each step a loss function is computed for each point $R$ of the current fraction as the number of minimal fractions becoming non-estimable when removing the point $R$. In formulae:
\begin{equation}\label{eq:app}
L(R) = \sum_{{\mathbf u}} \binom{n-\#\mathrm{supp}({\mathbf u})}{p-\#\mathrm{supp}({\mathbf u})}
\end{equation}
where the sum is taken over all the circuits $({\mathbf u})$ in the current fraction containing the point $R$. Notice that the formula in Equation \ref{eq:app} does not guarantee that the relevant minimal fractions are all distinct. The formula should be viewed as a first-order approximation of the inclusion-exclusion formula.

Therefore, the main steps of the algorithm are as follows:

\begin{enumerate}
\item Take the circuits $\mathcal{C}(X_\mathcal{D})$ of the candidate set under the given model, and a starting fraction ${\mathcal F}$ of a specified size $n$;

\item Select the circuits of $\mathcal{C}(X_\mathcal{D})$  with support on $p$ points or less. We denote this set of circuits by $\mathcal{C}^p(X_\mathcal{D})$;

\item Repeat until a {\it max-iter} number of iterations are performed:

\begin{enumerate}

\item Consider the circuits of $\mathcal{C}^p(X_\mathcal{D})$ which are contained in ${\mathcal F}$;

\item For each point $R$ in ${\mathcal F}$ compute its associated loss $L(R)$ as the (weighted) number of circuits which include $R$;

\item Take all the points with the highest loss and build up all possible pairs with one point not in ${\mathcal F}$. Make the exchange using the pair which reduces as much as possible the number of circuits contained in the fraction;

\item If no reduction is possible, then break.

\end{enumerate}

\end{enumerate}

In the simulation study below, the starting fraction ${\mathcal F}$ of a specified size $n$ is uniformly-at-random selected from all the subsets of size $n$ of the candidate set. In all the examples below the results are obtained on a sample of $1,000$ fractions.

As a first scenario, we describe the use of the algorithm on some examples where the candidate sets are full factorial designs. We consider four and five $2$-level factors with the main-effect model and two mixed-level cases. In both mixed-level cases we consider three factors, with 2,3 and 4 levels: in the first case we work with the main-effect model without interactions and in the second one with the main-effect model plus the interaction between the second and the third factor. The candidate sets are the $2^4$, the $2^5$ and the  $2 \times 3 \times 4$ full factorial designs respectively. The algorithm is used for finding robust fractions with different sizes. For each case the algorithm has been used starting from $1,000$ randomly selected fractions and using $20$ as the maximum number of iterations {\it max-iter}.  For each case in Tables \ref{std2a4}, \ref{std2a5},  \ref{std234}, and  \ref{std234plusone} some statistics concerning the robustness $r_B$ of the initial randomly selected designs and the difference $\delta$ between the final and the initial value of the robustness are reported. It is worth noting that in all but one cases the fifth percentile of $\delta$ is positive. It means that in 95$\%$ of the simulations the algorithm has been able to improve the initial design. The fifth percentile is negative for the $2^5$ design with 8 runs. In this case the twentieth percentile is positive $0.0357$  meaning the in 80$\%$ of the simulations the algorithm has been able to improve the initial design.


\begin{table} \label{std2a4}
\begin{center}
\begin{tabular}{|c|r r | r r r|} \hline
& \multicolumn{5}{|c|}{4 factors} \\ \hline
$n$ & $\bar{r}_B$ & med $r_B$ & $\bar{\delta}$ & med $\delta$ &  $\delta_{0.05}$ \\ \hline
8 &	0.687 &	0.679 &	0.135 &	0.143 &		0 \\
10 &	0.691 &	0.698 &	0.071 &	0.063 &		0.016 \\
12 &	0.69 &	0.689 &	0.032 &	0.033 &		0.005 \\
14 &	0.689 &	0.693 &	0.004 &	0 &		0 \\
 \hline
\end{tabular}
\caption{Mean ($\bar{r}_B$) and median (med $r_B$) of the robustness of the randomly selected fractions. Mean ($\bar{\delta}$), median (med $\delta$), and fifth percentile ($\delta_{0.05}$) of the differences between the final and the starting value of the robustness obtained in the $2^4$ design for various sizes of the fraction.}
\end{center}
\end{table}

\begin{table} \label{std2a5}
\begin{center}
\begin{tabular}{|c|r r | r r  r|} \hline
& \multicolumn{5}{|c|}{5 factors} \\ \hline
$n$ & $\bar{r}_B$ & med $r_B$ & $\bar{\delta}$ & med $\delta$ &   $\delta_{0.05}$ \\ \hline
8 &	0.615 &	0.607 &	0.198 &	0.179 &		-0.071 \\
10 &	0.613 &	0.624 &	0.156 &	0.138 &		0 \\
12 &	0.614 &	0.621 &	0.139 &	0.135 &		0.022 \\
14 &	0.613 &	0.618 &	0.109 &	0.107 &		0.034 \\
 \hline
\end{tabular}
\caption{Mean ($\bar{r}_B$) and median (med $r_B$) of the robustness of the randomly selected fractions. Mean ($\bar{\delta}$), median (med $\delta$), and fifth percentile ($\delta_{0.05}$) of the differences between the final and the starting value of the robustness obtained in the $2^5$ design for various sizes of the fraction.}
\end{center}
\end{table}

\begin{table} \label{std234}
\begin{center}
\begin{tabular}{|c|r r | r r  r|} \hline
& \multicolumn{5}{|c|}{ 2 $\times$ 3 $\times$ 4 no interaction} \\ \hline
$n$ & $\bar{r}_B$ & med $r_B$ & $\bar{\delta}$ & med $\delta$ &  $\delta_{0.05}$ \\ \hline
14 &	0.385 &	0.393 &	0.107 &	0.103 &		0.007 \\
16 &	0.385 &	0.394 &	0.077 &	0.07 &	 0.019 \\
18 &	0.383 &	0.387 &	0.045 &	0.041 &	0.009 \\
 \hline
\end{tabular}
\caption{Mean ($\bar{r}_B$) and median (med $r_B$) of the robustness of the randomly selected fractions. Mean ($\bar{\delta}$), median (med $\delta$), and fifth percentile ($\delta_{0.05}$) of the differences between the final and the starting value of the robustness obtained in the $2 \times 3 \times 4$ design (without first-order interaction) for various sizes of the fraction.}
\end{center}
\end{table}

\begin{table} \label{std234plusone}
\begin{center}
\begin{tabular}{|c|r r | r r  r|} \hline
& \multicolumn{5}{|c|}{ 2 $\times$ 3 $\times$ 4 with interaction} \\ \hline
$n$ & $\bar{r}_B$ & med $r_B$ & $\bar{\delta}$ & med $\delta$ &   $\delta_{0.05}$ \\ \hline
14 &	0.008 &	0 &	0.277 &	0.286 &		0.286 \\
16 &	0.01 &	0 &	0.047 &	0.057 &		0 \\
18 &	0.01 &	0 &	0.012 &	0.022 &		0 \\
20 &	0.01 &	0.013 &	0.003 &	0 &	0 \\
 \hline
\end{tabular}
\caption{Mean ($\bar{r}_B$) and median (med $r_B$) of the robustness of the randomly selected fractions. Mean ($\bar{\delta}$), median (med $\delta$), and fifth percentile ($\delta_{0.05}$) of the differences between the final and the starting value of the robustness obtained in the $2 \times 3 \times 4$ design (with first-order interaction)  for various sizes of the fraction.}
\end{center}
\end{table}

In Figure \ref{a510runs} the $1,000$ pairs of robustness of the starting fraction and robustness of the final fraction are reported in the case of $5$ factors and $n=10$ as pre-specified size of the fraction.

\begin{figure}
\begin{center}
\includegraphics[width=7cm]{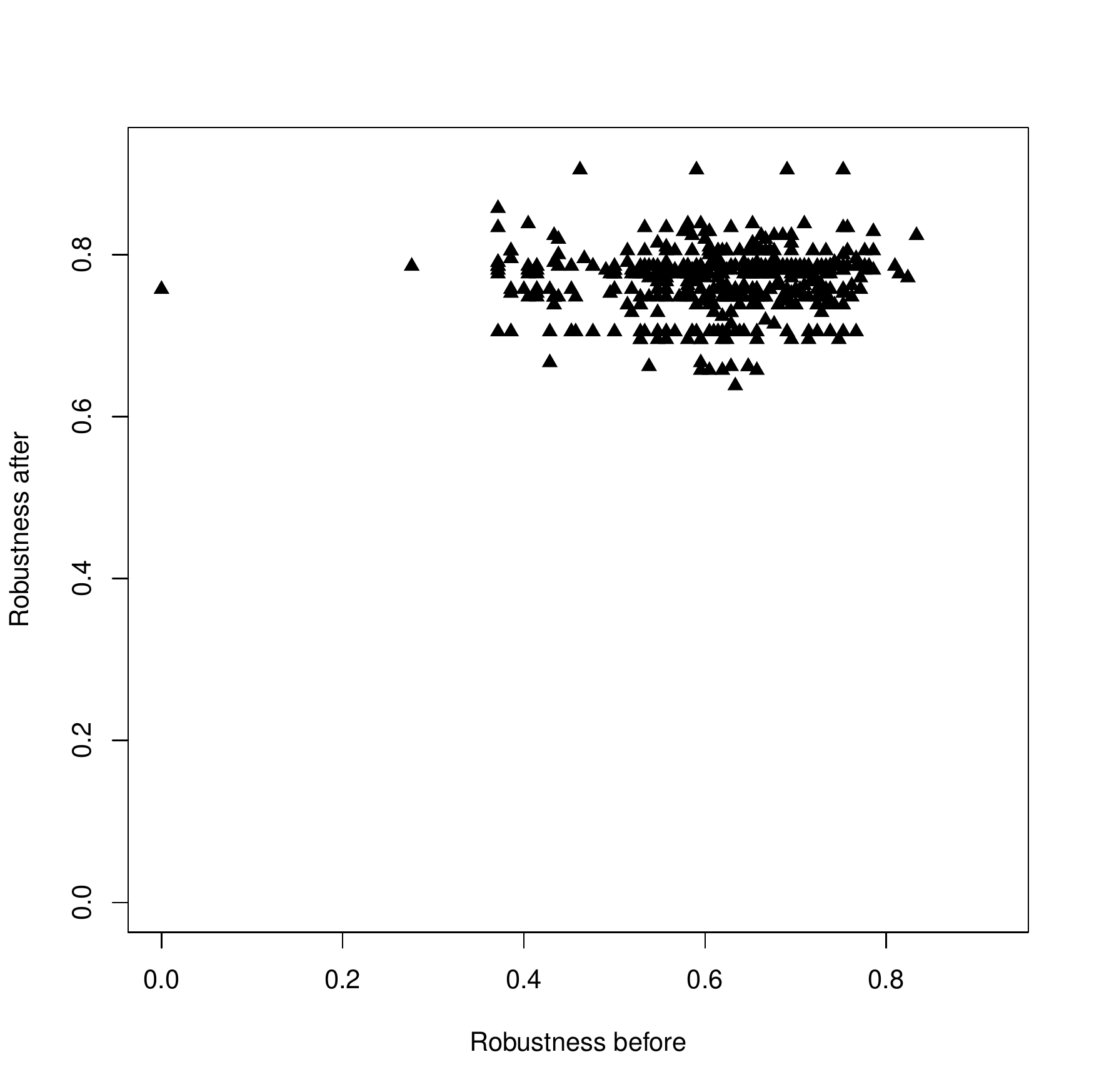}
\end{center}
\caption{Robustness of the input and the output fraction for the model $2^5$ with main effects, size $n=10$.} \label{a510runs}
\end{figure}


As a second scenario, we consider the problem of finding robust subsets of a given Orthogonal Array. In these cases, the problem on  the full-factorial design is not feasible, but the algorithm can still be used thanks to the properties of the circuits for sub-fractions discussed in Section \ref{sect:circandrob}.

To illustrate a first example in this scenario, let us consider the Orthogonal Array with strength 3 in \cite{tonchev:89}. It consist of 40 runs of a $2^{20}$ full factorial design, and it also reported in the collection neilsloane.com/oadir, see \cite{oadir}. Under the first-order model, the circuits contained in the Orthogonal Array are 190, all with support on $4$ points. We have performed a simulations by running the algorithm for finding robust fractions with 22 runs from a random starting fractions. One easily finds that in all $1,000$ replicates a fraction with $r=0.1818$ is generated, while the starting random fraction has $r=0$ in all but one cases. The same holds when finding fractions with 23 runs, where in all cases the algorithm selects a fraction with robustness $r=0.0474$. In this example all the moves are basic moves, so the simplified version of the algorithm does not modify the computations. Since the robustness of the output fractions is constant, we can recover the number of saturated fractions in each case. So, for 22 runs we have 4 saturated fractions, for 23 runs we have 12 saturated fractions, for 24 runs we have 32 saturated fractions. 

As another example, we move to a non-binary example. Consider the best GWLP 3 Orthogonal Array with 27 runs in the $3^4$ full factorial design, see \cite{eendebak:sito}. In this case, under the main-effect and first order interaction model, there are $58,113$ circuits in the relevant Orthogonal Array, $36,045$ of which can be discarded since their support is on 10 points. Here, only 81 moves are basic moves. In Figure \ref{OAnonbin} the $1,000$ pairs of robustness of the starting fraction and robustness of the final fraction are reported when fractions of size $n=12$ are considered. In the case $n=12$, while the mean robustness of the starting fraction is $\overline r_B= 0.3197$, the mean robustness of the final fraction is $\overline r_A=0.4567$. 

\begin{figure}
\begin{center}
\includegraphics[width=7cm]{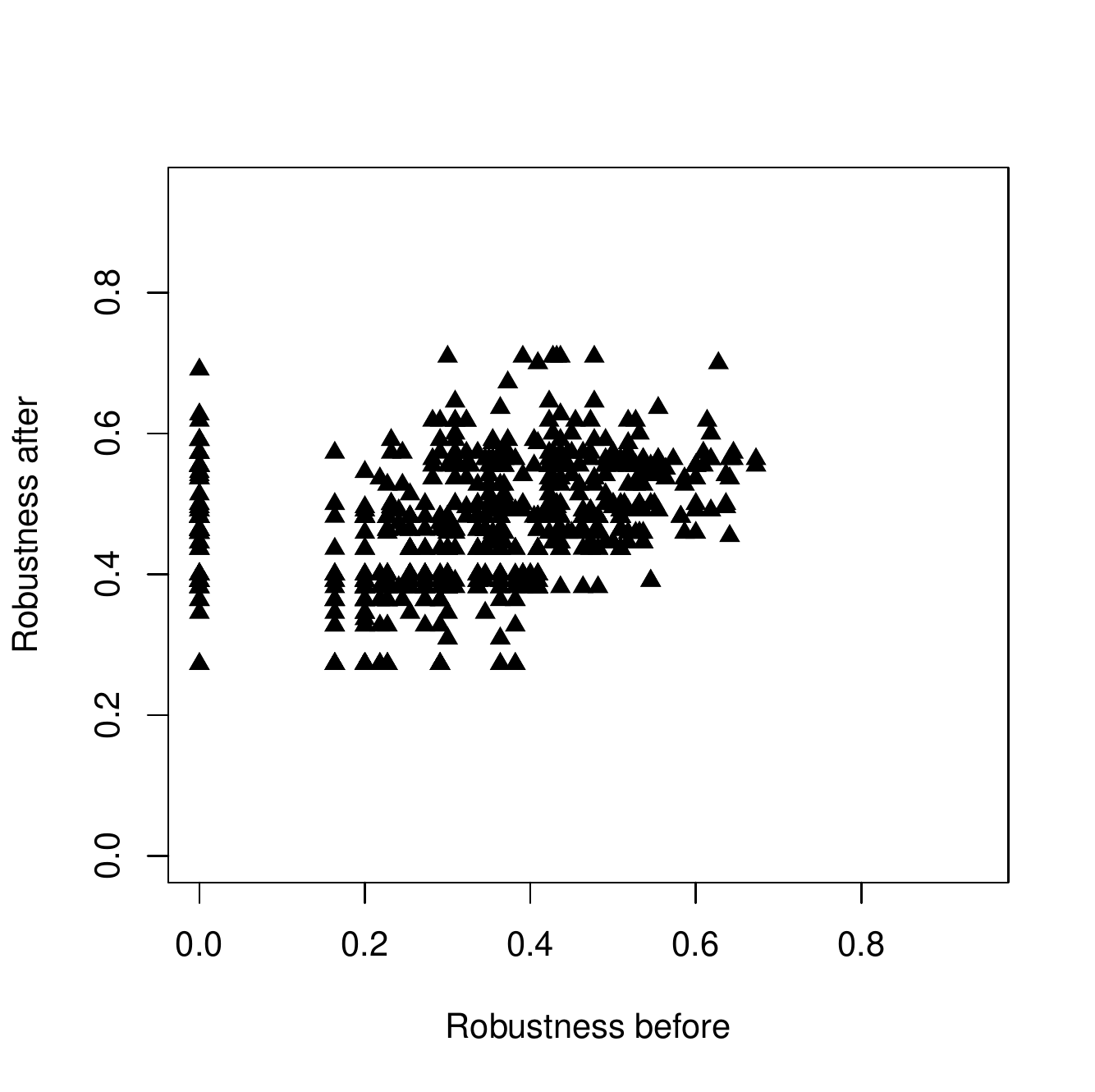}
\end{center}
\caption{Robustness of the input and the output fraction for the subsets of size $n=12$ from the best GWLP 3 Orthogonal Array with 27 runs in the $3^4$ full factorial design.} \label{OAnonbin}
\end{figure}

\subsection{Computational remarks}

As we have already noticed, the proposed algorithm is a first order approximation of the inclusion-exclusion formula. In case of ties, both for the run to be excluded and the candidate new run, the algorithm seeks for the optimal exchange, so that it stops in few steps. In all our simulations, the parameter {\it max-iter} is set to 20 but in most cases 3 or 4 iterations are enough to reach a stationary point.

A first-order approximation of the inclusion-exclusion formula can be not accurate when there is a large number of circuits contained in the proposed fraction. Indeed, the algorithm computes for each point the (weighted) number of circuits containing the point, and this number is taken as an estimate of the number of minimal fractions which would become estimable if the point is removed. However, several circuits can pertain to the same minimal fraction. For this reason, when the number of desired runs becomes large, one can consider a reduced version of the algorithm taking into account only the circuits with minimal support. The algorithm is performed as above, simply by taking the set of circuits with minimal support $\mathcal{C}^{min}(X_\mathcal{D})$ instead of $\mathcal{C}(X_\mathcal{D})$.

%
%
%
%
%
%
%
%
%

For each case in Tables \ref{sim2a4}, \ref{sim2a5},  \ref{sim234}, and  \ref{sim234plusone} some statistics concerning the robustness $r_B$ of the initial randomly selected designs and the difference $\delta$ between the final and the initial value of the robustness are reported. The reduced version of the algorithm performs extremely well.  In all cases the fifth percentile of $\delta$ is positive (or null in some cases). It means that in 95$\%$ of the simulations the reduced algorithm has been able to improve the initial design. For the $2^4$ and the $2^5$ cases, fractions with $10$ runs, the robustness before and after the reduced algorithm is plotted in Fig.~\ref{fig:10runsred} for $1,000$ randomly generated starting fractions.

\begin{table} \label{sim2a4}
\begin{center}
\begin{tabular}{|c|r r | r r  r|} \hline
& \multicolumn{5}{|c|}{4 factors} \\ \hline
$n$ & $\bar{r}_B$ & med $r_B$ & $\bar{\delta}$ & med $\delta$ &  $\delta_{0.05}$ \\ \hline
8 &	0.695 &	0.714 &	0.109 &	0.107 &		0 \\
10 &	0.688 &	0.698 &	0.063 &	0.063 &		0 \\
12 &	0.688 &	0.689 &	0.032 &	0.033 &		0.005 \\
14 &	0.689 &	0.693 &	0.005 &	0 &	0 \\
 \hline
\end{tabular}
\caption{Mean ($\bar{r}_B$) and median (med $r_B$) of the robustness of the randomly selected fractions. Mean ($\bar{\delta}$), median (med $\delta$), and fifth percentile ($\delta_{0.05}$) of the differences between the final and the starting value of the robustness obtained in the $2^4$ design for various sizes of the fraction. Reduced algorithm.}
\end{center}
\end{table}

\begin{table} \label{sim2a5}
\begin{center}
\begin{tabular}{|c|r r | r r  r|} \hline
& \multicolumn{5}{|c|}{5 factors} \\ \hline
$n$ & $\bar{r}_B$ & med $r_B$ & $\bar{\delta}$ & med $\delta$ &  $\delta_{0.05}$ \\ \hline
8 &	0.605 &	0.607 &	0.231 &	0.214 &	0 \\
10 &	0.613 &	0.624 &	0.244 &	0.229 &	0.095 \\
12 &	0.612 &	0.616 &	0.257 &	0.259 &	0.12 \\
14 &	0.613 &	0.619 &	0.147 &	0.141 &	0.075 \\
 \hline
\end{tabular}
\caption{Mean ($\bar{r}_B$) and median (med $r_B$) of the robustness of the randomly selected fractions. Mean ($\bar{\delta}$), median (med $\delta$), and fifth percentile ($\delta_{0.05}$) of the differences between the final and the starting value of the robustness obtained in the $2^5$ design for various sizes of the fraction. Reduced algorithm.}
\end{center}
\end{table}

\begin{table} \label{sim234}
\begin{center}
\begin{tabular}{|c|r r | r r  r|} \hline
& \multicolumn{5}{|c|}{ 2 $\times$ 3 $\times$ 4 no interaction} \\ \hline
$n$ & $\bar{r}_B$ & med $r_B$ & $\bar{\delta}$ & med $\delta$ &  $\delta_{0.05}$ \\ \hline
14 &	0.382 &	0.392 &	0.148 &	0.138 &		0.056 \\
16 &	0.384 &	0.389 &	0.088 &	0.081 &		0.033 \\
18 &	0.383 &	0.387 &	0.047 &	0.044 &		0.012 \\
 \hline
\end{tabular}
\caption{Mean ($\bar{r}_B$) and median (med $r_B$) of the robustness of the randomly selected fractions. Mean ($\bar{\delta}$), median (med $\delta$), and fifth percentile ($\delta_{0.05}$) of the differences between the final and the starting value of the robustness obtained in the $2 \times 3 \times 4$ design (without first-order interaction) for various sizes of the fraction. Reduced algorithm.}
\end{center}
\end{table}

\begin{table} \label{sim234plusone}
\begin{center}
\begin{tabular}{|c|r r | r r  r|} \hline
& \multicolumn{5}{|c|}{ 2 $\times$ 3 $\times$ 4 with interaction} \\ \hline
$n$ & $\bar{r}_B$ & med $r_B$ & $\bar{\delta}$ & med $\delta$ &   $\delta_{0.05}$ \\ \hline
14 &	0.011 &	0 &	0.274 &	0.286 &	0.286 \\
16 &	0.01 &	0 &	0.047 &	0.057 &	0 \\
18 &	0.01 &	0 &	0.013 &	0.022 &	0 \\
20 &	0.01 &	0.013 &	0.004 &	0 &	0 \\
 \hline
\end{tabular}
\caption{Mean ($\bar{r}_B$) and median (med $r_B$) of the robustness of the randomly selected fractions. Mean ($\bar{\delta}$), median (med $\delta$), and fifth percentile ($\delta_{0.05}$) of the differences between the final and the starting value of the robustness obtained in the $2 \times 3 \times 4$ design (with first-order interaction)  for various sizes of the fraction. Reduced algorithm.}
\end{center}
\end{table}

\begin{figure}
\begin{center}
\begin{tabular}{ccc}
\includegraphics[width=7cm]{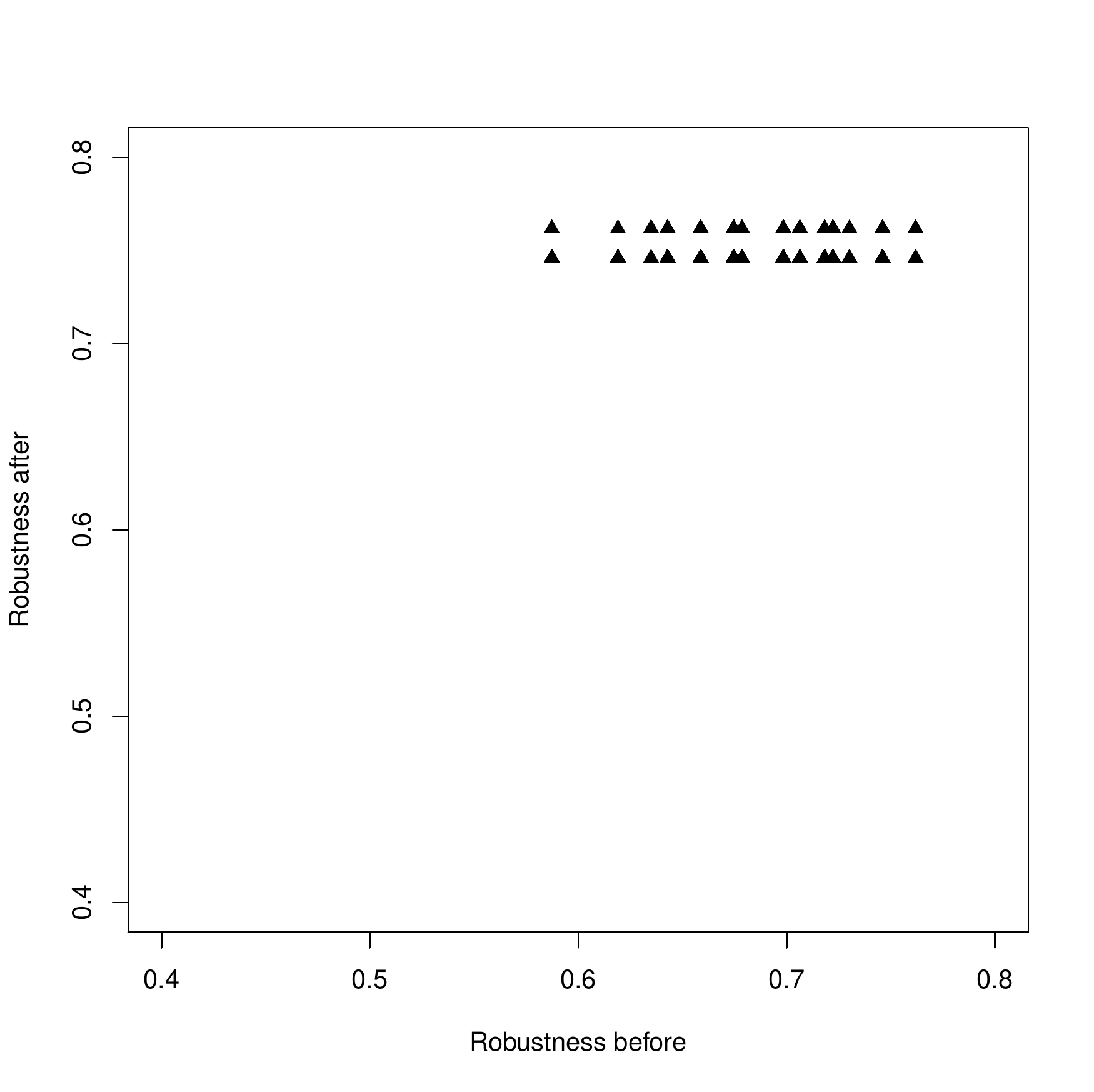}
& $ $ &
\includegraphics[width=7cm]{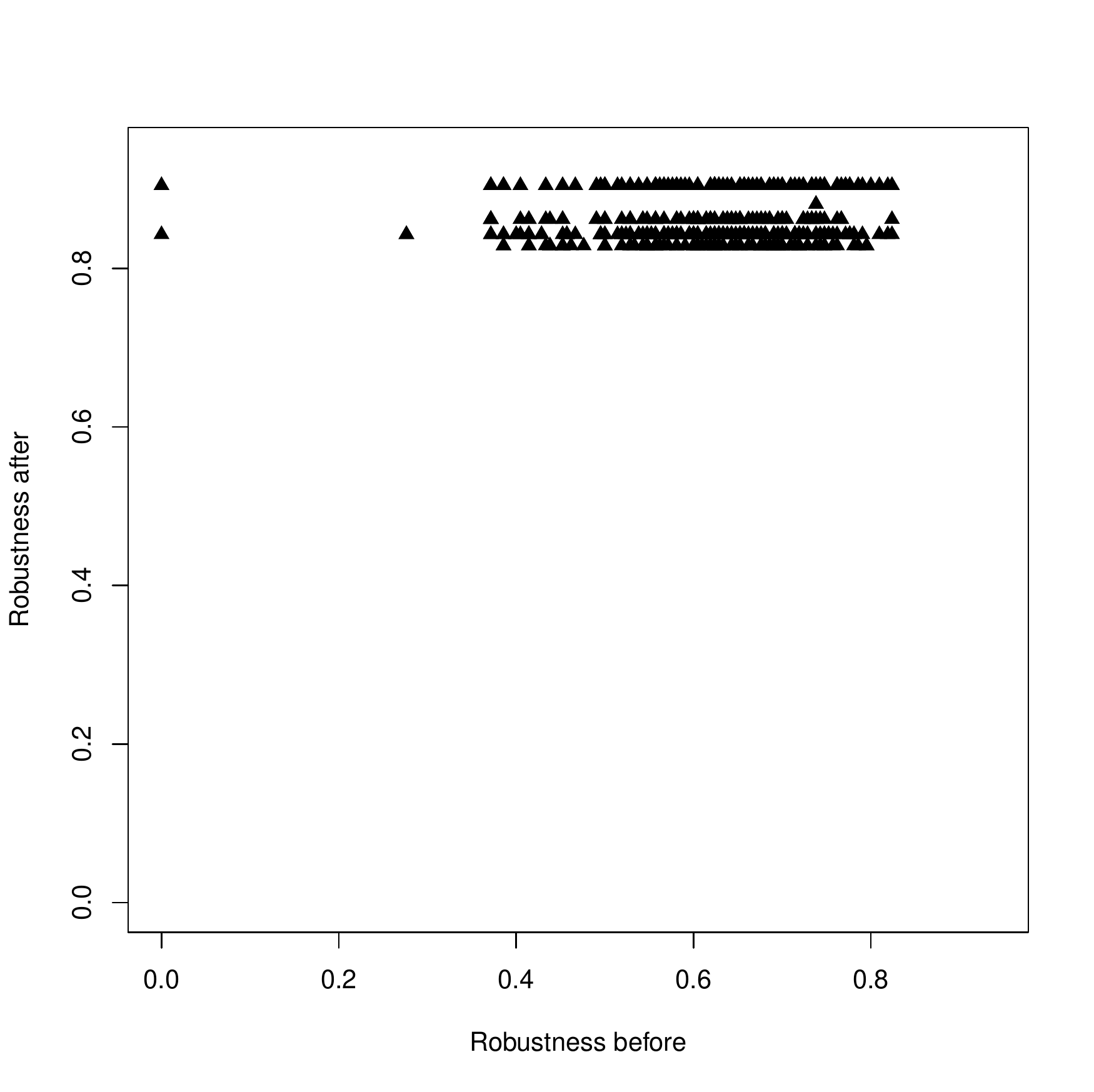}
\end{tabular}
\end{center}
\caption{Robustness of the input and the output fraction for the subsets of size $n=10$ for the $2^4$ (left) and $2^4$ (right) models with main effects. Reduced algorithm.} \label{fig:10runsred}
\end{figure}

The results here confirm that for small $n$ the reduced algorithm has worse performances, but it outperforms the complete algorithm when the design size $n$ increases. The simplified version can be used also in the case of large designs, since in most cases the circuits with minimal support can be defined theoretically, without computations.

Also, note that the algorithm does not need the computation of the robustness at each step (as for instance in a standard exchange algorithm). As a consequence, the execution time is  less than 1 sec per fraction in all the examples.

Finally, since there are several local minima, the algorithm can be put into a standard simulated annealing. We have not explicitly considered this option in our algorithm since the main aim of this paper is to highlight the connections between robustness and the geometry of the fraction studied though the circuits.

\subsection{Robustness and D-optimality}

From our previous results in Section \ref{sect:definition}, we know that robustness is equivalent to $D$-optimality in the case of totally unimodular model matrices. Thus, intuition suggests to use a $D$-optimal fraction as a starting point of our algorithm also in the general case. However, some simple simulations show that in for general model matrices $D$-optimal fractions are far from being also robust. To illustrate this feature, we show in Figure \ref{Doptvsrob} the scatterplot of the $D$-efficiency versus the robustness of a sample of fractions (including the $D$-optimal one) for the model with main effects in the $2^4$ factorial design. It is immediate to see that in both cases the $D$-optimal fraction has a low value of robustness.

\begin{figure}
\begin{center}
\begin{tabular}{ccc}
\includegraphics[width=7cm]{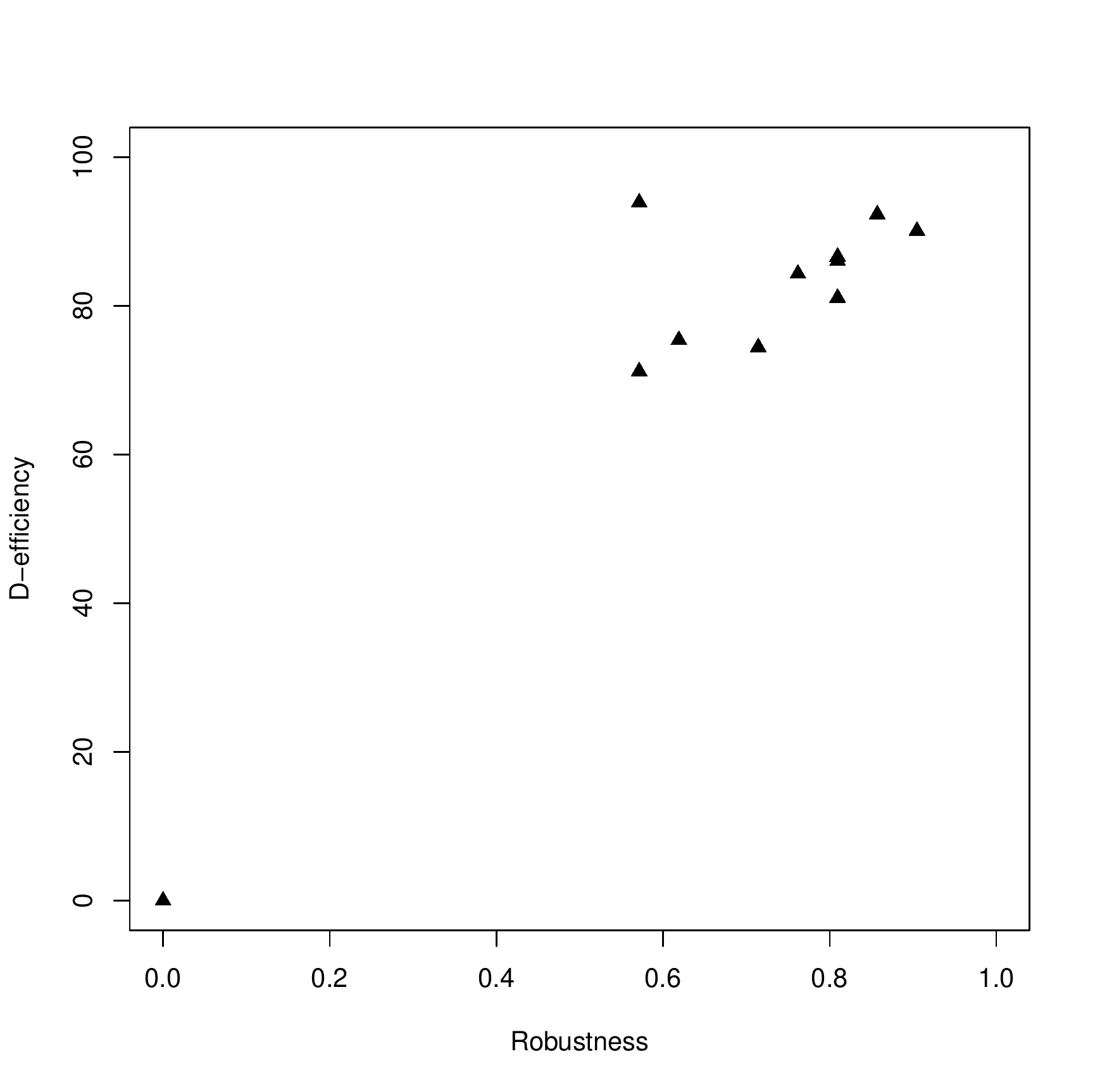} & $ $ & \includegraphics[width=7cm]{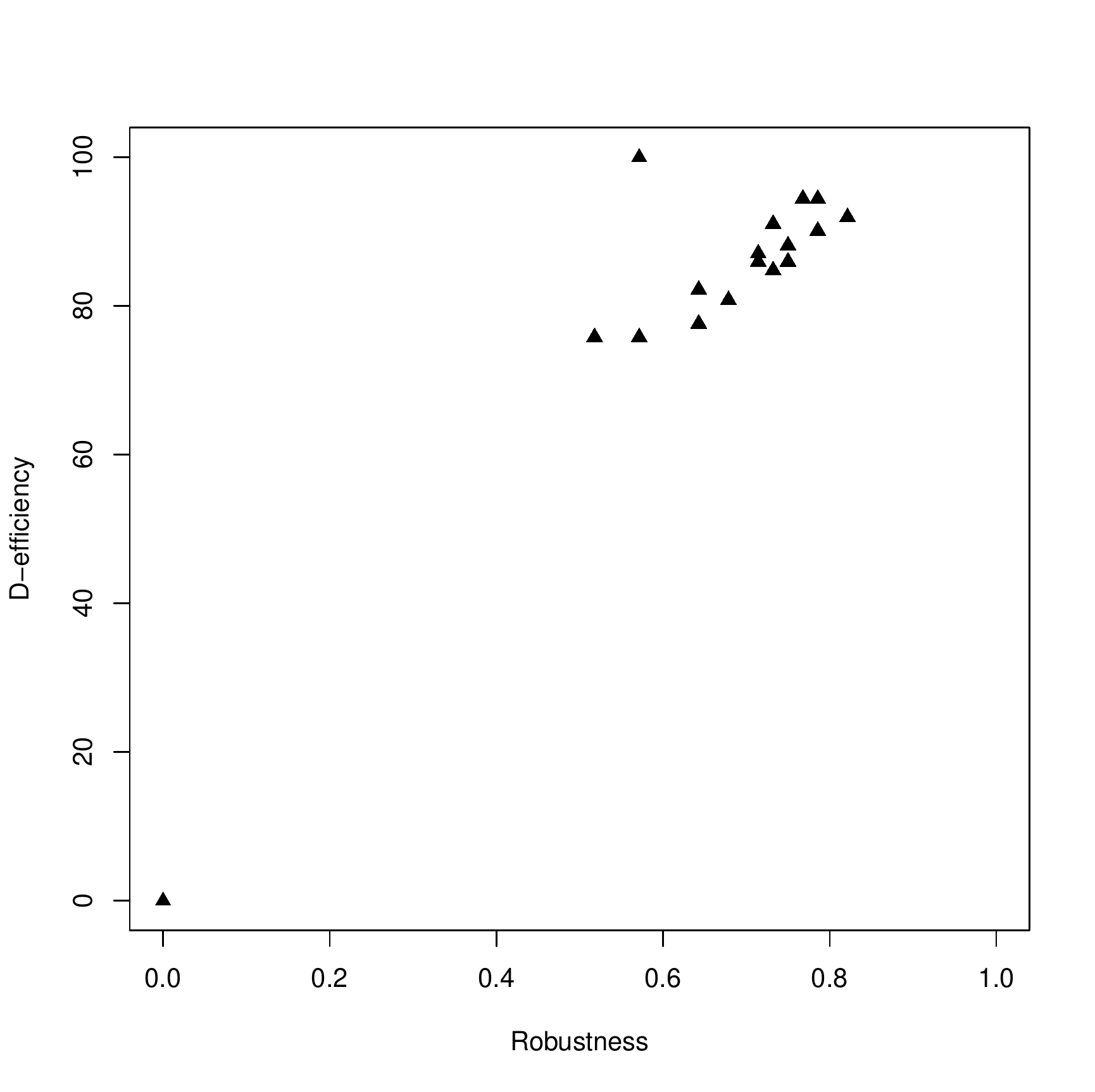}
\end{tabular}
\end{center}
\caption{$D$-efficiency versus robustness of fractions in the $2^4$ model with main effects ($n=7$ on the left, $n=8$ on the right).} \label{Doptvsrob}
\end{figure}

\section{Final remarks and the case of supersaturated designs} \label{sect:final}

In this work, we have introduced an algorithm for finding robust fractions (i.e., subsets of a candidate set of design points) using the combinatorial notion of circuit basis, and thus highlighting the geometric nature of the problem. We have shown through several examples that the proposed algorithm is effective, and can be applied also in the case of problems with moderate size, by exploiting the properties of the circuit basis. In this concluding section, we introduce the analysis of supersaturated models by means of the circuit basis. The detailed analysis of this problem needs some more theory and falls outside the scope of the present paper, but nevertheless we aim at introducing the basic facts needed to proceed in that direction.

Now, fix a fraction ${\mathcal F}$ and consider the design matrix partitioned by columns into two sub-matrices $X^{(1)}$ and $X^{(2)}$:
\begin{equation*}
X^{(2)} = \left( X^{(1)} \ | \ X^{(a)} \right) \, .
\end{equation*}

Here $X^{(2)}$ and $X^{(1)}$ are two nested models on the same fraction. In this case the connections between the circuit bases for $X^{(2)}$ and for $X^{(2)}$ are less simple, but the following result holds.

\begin{propo} \label{prop:ssd}
If $X^{(2)}$ is a model matrix, and $X^{(1)}$ is a model matrix obtained from $X^{(2)}$ by removing columns as above, then:
\begin{enumerate}
\item From $X^{(1)}$ to $X^{(2)}$:
\begin{equation*}
 \mathbf{u} \in {\mathcal C}({X^{(1)}}) \ \mbox{ s.t. } \ \mathbf{u} \in \ker(X^{(2)})  \Rightarrow  \mathbf{u} \in {\mathcal C}({X^{(2)}}) \, .
\end{equation*}

\item From $X^{(2)}$ to $X^{(1)}$:
\begin{equation*}
 \mathbf{u} \in {\mathcal C}({X^{(2)}}) \Rightarrow  \mathbf{u} \in {\mathcal C}({X^{(1)}})
\end{equation*}
or there exists $\mathbf{v} \in {\mathcal C}({X^{(1)}})$ with ${\rm supp}(\mathbf{v}) \subset {\rm supp}(\mathbf{u})$. In such a case, $\mathbf{v}^tX^{(a)} \ne 0$.
\end{enumerate}
\end{propo}

\begin{proof}
Let $u$ be a circuit in ${\mathcal C}({\mathcal F},X_1)$ and $ u \in \ker(X_2)$. The fact that $u$ is support-minimal for $X_2$ follows immediately by contradiction.
\end{proof}

The analysis of subfractions by means of Prop.~\ref{prop:ssd} is less easy, but some aid from the circuit bases still survives. In fact, for minimal fractions with $p$ runs either $R({\mathcal F})=1$ or $R({\mathcal F})=0$. Moreover, for saturated fractions $X_p$ is non singular so ${\mathcal C}({X_p})$ is the empty set.

Now we can consider from a saturated fraction sub-fractions with $(p-1)$ runs on models with $(p-1)$ parameters, and we look at maximizing the number of non-singular $X_{(p-1)\times(p-1)}$ matrices.

For saturated fractions
\[
\dim \ker X_{p} = 0 \, .
\]
Thus, removing one parameter (column)
\[
\dim \ker X_{p\times(p-1)} = 1 \, .
\]

Therefore the generator of $\ker X_{p\times(p-1)}$ is the unique circuit.

The basic idea here is to compute the $p$ circuits for all possible $ X_{p\times(p-1)}$ matrices. The zeros in these circuits, gives an index of robustness for saturated fractions and helps in finding subfractions. In fact, following our theory, the number of zeros in these circuits corresponds exactly to the number of singular matrices $X_{(p-1)\times(p-1)}$. Let us illustrate this fact with an example.

\begin{ex}
Consider the 8-run Plackett-Burman design for the $2^7$ problem. With the usual $\pm$ notation, it is diplayed in Fig.~\ref{fig:pb}, together with the corresponding model matrix.

\begin{figure}
\begin{equation*}
\begin{pmatrix}
+ & + & + & + & + & + & + \\
+ & + & - & - & - & - & + \\
+ & - & + & + & - & - & - \\
+ & - & - & - & + & + & - \\
- & + & + & - & + & - & - \\
- & + & - & + & - & + & - \\
- & - & + & - & - & + & + \\
- & - & - & + & + & - & + \end{pmatrix}
\end{equation*}
\caption{The 8-run Plackett-Burman design.} \label{fig:pb}
\end{figure}

The 8 circuits obtained by deletion of one column from $X_{\mathcal F}$ are reported in Fig.~\ref{fig:pbris}. We see that all the entries of the 8 circuits are non-zero, and thus all 7-points sub-fractions are estimable in all models with one removed column, providing the optimality of the Plackett-Burman design in terms of robustness.

\begin{figure}
\begin{equation*}
\begin{matrix*}[r]
 1 & 1 & 1 & 1 & 1 & 1 & 1 & 1 \\
 1 & 1 & 1 & 1 & -1 & -1 & -1 & -1 \\
 1 & 1 & -1 & -1 & 1 & 1 & -1 & -1 \\
 1 & -1 & 1 & -1 & 1 & -1 & 1 & -1 \\
 1 & -1 & 1 & -1 & -1 & 1 & -1 & 1 \\
 1 & -1 & -1 & 1 & 1 & -1 & -1 & 1 \\
 1 & -1 & -1 & 1 & -1 & 1 & 1 & -1 \\
 1 & 1 & -1 & -1 & -1 & -1 & 1 & 1 \end{matrix*}
\end{equation*}
\caption{The 8 circuits obtained by deletion of one column from the 8-run Plackett-Burman design.} \label{fig:pbris}
\end{figure}
\end{ex}

This example shows that the connections between the statistical properties of a design and its geometry are not limited to the algorithm for robust fractions introduced in this paper. There are several possible new applications of the circuit basis for the study of the structure of a design, as for instance supersaturated fractions, optimal designs, and randomization for treatment allocation.


\end{document}